\documentclass[a4paper,12pt,oneside]{article}

\usepackage{amsmath}
\usepackage{amsthm}
\usepackage{amssymb}
\usepackage{booktabs}

\newtheorem{theorem}{Theorem}
\newtheorem{lemma}{Lemma}
\newtheorem{remark}{Remark}
\newtheorem{definition}{Definition}

\title{Selection and identity rules for subductions of type A quantum Iwahori-Hecke algebras}
\author{Vincenzo Chilla}
\date{}

\begin{document}
\maketitle
\begin{abstract}
This paper is concerned with the subduction problem of type A quantum Iwahori-Hecke algebras $\mathbb{C} \mathbf{H}(\mathfrak{S}_f,q^2)$ with a real deformation parameter $q$, i.e. the problem of decomposing irreducible representations of such algebras as direct sum of irreducible representations of the subalgebras $\mathbb{C}\mathbf{H}(\mathfrak{S}_{f_1}, q^2) \times \mathbb{C}\mathbf{H}(\mathfrak{S}_{f_2}, q^2)$, with $f_1 + f_2 = f$. After giving a suitable combinatorial description for the subduction issue, we provide a selection rule, based on the Richardson-Littlewood criterion, which allows to determine the vanishing coupling coefficients between standard basis vectors for such representations, and we also present an equivariance condition for the subduction coefficients. Such results extend those ones corresponding to the subduction problem in symmetric group algebras $\mathbb{C}\mathfrak{S}_f \downarrow \mathbb{C}\mathfrak{S}_{f_1} \times \mathbb{C} \mathfrak{S}_{f_2}$ which are obtained by $q$ approaching the value $1$.
\end{abstract}
\newpage
\section{Introduction}
Introduced indipendently by Drinfeld and Jimbo~\cite{jimbo} in 1985, quantum groups soon appeared in close connection with a quantum mechanical problem in statistical mechanics: the quantum Yang-Baxter equation~\cite{baxter}. It was then realized that quantum groups had far-reaching applications in theoretical physics (e.g. conformal and quantum field theories), knot theory, and virtually many other areas of mathematics and mathematical physics. The notion of a quantum group is also closely related to the study of integrable dynamical systems~\cite{felder}, from which the concept of a Poisson-Lie group emerged, and to the classical Weyl-Moyal quantization~\cite{casalbuoni}.
\par
The type A Iwahori-Hecke algebras are special quantum group algebras realizations which arise naturally in the following setting: let $\mathfrak{U}$ be a quantum group corresponding to a finite dimensional complex simple Lie algebra of type A, and let V be the irreducible representation of $\mathfrak{U}$ corresponding to the fundamental weight $\omega$. Then the centralizer algebra $\mathfrak{Z}_m= End_{\mathfrak{U}}(V^{\otimes m})$ is isomorphic to a quotient of Iwahori-Hecke algebras~\cite{jimbo1}. Thus a quantum version of the classical Schur-Weyl duality~\cite{schur} between symmetric group algebras and unitary groups can be established. 
\par 
The quantum Schur-Weyl duality provides an efficient approach to the so-called \emph{Wigner-Racah calculus} for quantum unitary groups~\cite{pan}, i.e. the issue of determining coupling and recoupling coefficients (Clebsch-Gordan coeffcients, $3j$ symbols, $6j$ symbols, $\ldots$) of irreducible representations of such groups. On the other hand, the quantum Racah-Wigner algebra of unitary groups plays a fundamental role in the representation theory of quantum Lie algebras and its importance is represented by the connection, for example, with the study of many models in statistical mechanics~\cite{pasquier} and with the determination of new knot invariants~\cite{kassel}. In particular, $6j$ symbols are very useful and they were discussed by many authors using different methods~\cite{kachurik,bo,rajeswari}. Here, we would like to emphasize the works of Kramer~\cite{kramer} and Chen et al~\cite{chen}. They first used the Schur-Weyl duality relation between symmetric group algebras $\mathbb{C}\mathfrak{S}_f$ and unitary groups $U(n)$, which enable them to derive $U(n)$ Racah coefficients from \emph{subduction coefficients} of centralizer algebras. Such ideas was then extended to other classical Lie groups~\cite{pan1,pan2} and their quantum deformations~\cite{dai1,dai2}.
\par
Subduction coefficients for symmetric groups were first introduced in 1953 by Elliot {\it et al}~\cite{Elliot} to describe the states of a physical system with $f$ identical particles as composed of two subsystems with $f_1$ and $f_2$ particles respectively ($f_1 + f_2 = f$). Since Elliot {\it et al} (1953), many techniques have been proposed for calculating the subduction coefficients, but the investigation is until now surely incomplete. The main goal to give explicit and general closed algebraic formulas has not been achieved. Only some special cases have been solved~\cite{Kaplan2, Rao, McAven}. There are numerical methods~\cite{Horie, Kaplan1, Chen1} which are used to approach the issue, but no insight into the structure of the trasformation coefficients can be obtained. Another key outstanding problem is to resolve multiplicity separations~\cite{Butler} in a systematic manner, indicating a consistent choice of the indipendent phases and free factors. In~\cite{McAven1, McAven2} a breakthrough was made about this; however, the authors abandon the aim to obtain an algebraic solution and prefer a combinatorial recipe. 
\par
In~\cite{Chilla}, we came back to a general and algebraic approach to the subduction problem in symmetric groups and we analyzed in detail the \emph{linear equation method}, an efficient tool for the determination of the subduction coefficients as solution of a linear system. We showed that such a system, which  is constituted by a complicated primal structure of dependent linear equations, can be simplified by choosing a minimal set of sufficient equations related to the combinatorial concept of \emph{subduction graph}. Thus we could outline an improved and general algorithm to solve the subduction problem in symmetric groups by a graph searching process. Furthermore, we proposed a general form for the subduction coefficients resulting from the only requirement of orthonormality and we saw that the multiplicity separation can be described in terms of the Sylvester matrix of a suitable positive defined quadratic form describing the scalar product in the subduction space. Then, we was able to link the freedom in fixing the multiplicity separation to the freedom deriving from the choice of such a Sylvester matrix. 
\par
This paper represents a natural continuation and one of the possible generalizations of~\cite{Chilla}. Here, we focus on the subduction problem involving semisimple representations of type A quantum Iwahori-Hecke algebras and we analize the combinatorial structure of the reduction system arisen in the changing from standard to non-standard bases. So, we get a more insight into the structure of the solution for such a system, providing a \emph{selection rule} and an \emph{identity rule} for the subduction coefficients which allow further on reducing the number of unknowns equations. Consequently, we can drastically increment the dimension for the involved irreducible representations and find solutions for higher multiplicity cases.     
\section{ Type A Iwahori-Hecke algebras and their irreducible representations}
The Iwahori-Hecke algebra of type A~\cite{geck}, denoted here by $\mathbb{C} \mathbf{H}(\mathfrak{S}_f, q^2)$, is the algebra over $\mathbb{C}(q)$, the field of rational functions over $\mathbb{C}$, generated by $1$, $g_1$, $\ldots$, $g_{f-1}$ subject to the relations
\begin{subequations} \label{gen}
\begin{align}
g_i g_{i+1} g_i & = g_{i+1} g_i g_{i+1} \label{gen1}  \\
g_i g_j & = g_j g_i \ \ \ \ \ \ \ \ \ \ \ \ \ \ \  \ \ \  \ \ \textmd{if \phantom{a}  $|i-j| \geq 2$} \label{gen2}\\
g_i^2 & = (q - q^{-1}) g_i + 1 \label{gen3}
\end{align}
\end{subequations}
Therefore, $\mathbb{C} \mathbf{H}(\mathfrak{S}_f, q^2)$ can be considered as a quantum deformation of the symmetric group algebra $\mathbb{C}\mathfrak{S}_f$ which is defined from the previous relations by setting $q=1$ in equation (\ref{gen3}). In that case, each generator $g_i$ represents the elementary transposition which exchanges $i$ and $i+1$.
\par 
It is known~\cite{wenzl} that $\mathbb{C} \mathbf{H}(\mathfrak{S}_f, q^2)$ is not semisimple only for $q$ being a primitive $k^{th}$ root of unity, with $k=2, 3, \ldots, f$. In this paper, we choose $q$ to be a real parameter. Therefore $\mathbb{C} \mathbf{H}(\mathfrak{S}_f, q^2)$ is always supposed to be semisimple. \\ Whenever $\mathbb{C} \mathbf{H}(\mathfrak{S}_f, q^2)$ is semisimple, each irreducible representation $[\lambda]$, up to conjugacy, is labelled by a Young diagram $\lambda$ with $f$ boxes~\cite{leduc}. Given an irreducible representation $[\lambda]$, the standard basis vactors are defined by the multiplicity-free Gelfand-Tzetlin chain $\mathbb{C} \mathbf{H}(\mathfrak{S}_f, q^2) \subset \mathbb{C} \mathbf{H}(\mathfrak{S}_{f-1}, q^2) \subset \ldots \subset \mathbb{C} \mathbf{H}(\mathfrak{S}_1, q^2)$ and, as in the case of the symmetric group algebras, they are orthonormal vectors labelled by standard Young tableaux with diagram $\lambda$. Following a conventional notation, we denote by $|\lambda; m \rangle_q$ the standard basis vector associated to the standard Young tableau $m$.
\par 
The irreducible representation $[\lambda]$, in the standard basis (sometimes called \emph{Yamanouchi basis}), is given by~\cite{leduc}
\begin{equation}
g_i | \lambda; m\rangle_q = \frac{q^{d_i(m)}}{[d_i(m)]_q} |\lambda; m \rangle_q + \beta_{i,q}(m) |\lambda; g_i(m)\rangle_q
\label{actstd}
\end{equation}
where, for a fixed real number $x$, the \emph{quantum number} $[x]_q$ is defined by
\begin{equation}
[x]_q=\frac{q^x-q^{-x}}{q-q^{-1}},
\nonumber
\end{equation}
$\beta_{i,q}(m)$ is given by
\begin{equation}
\beta_{i,q}(m) = \left( 1- \frac{1}{[d_i(m)]_q^2} \right)^{1/2}
\nonumber
\end{equation}
and $d_i(m)$ represents the usual \emph{axial distance} from $i$ and $i+1$ in $m$. Finally, $g_i(m)$ represents the action of the generators $g_i$ on the standard Young tableaux. It is defined as follows: $g_i(m)$ is the tableau obtained by $m$ interchanging $i$ and $i+1$, if that one is another standard Young tableau; else we set $g_i(m)=m$. It is interesting to observe that, with the previous definition, $g_i(m)\neq m$ if and only if $|d_i(m)|\neq 1$. 
\section{Non-standard bases and subduction equations}
Beyond the standard basis, a non-standard basis, sometimes called \emph{split basis}, can be usefully choosen for the irreducible representation $[\lambda]$ of $\mathbb{C} \mathbf{H}(\mathfrak{S}_f, q^2)$. By definition, such a basis breaks $[\lambda]$ (which is, in general, a reducible representation of the direct product subalgebra  $\mathbb{C}\mathbf{H}(\mathfrak{S}_{f_1}, q^2) \times \mathbb{C}\mathbf{H}(\mathfrak{S}_{f_2}, q^2)$, with $f_1 + f_2 = f$) in a block-diagonal form
\begin{equation}
[\lambda]=\bigoplus_{\lambda_1, \lambda_2} \{\lambda; \lambda_1, \lambda_2 \} [\lambda_1] \otimes [\lambda_2],
\nonumber
\end{equation}
where $[\lambda_1]$ and $[\lambda_2]$ denote irreducible representations of $\mathbb{C}\mathbf{H}(\mathfrak{S}_{f_1}, q^2)$ and $\mathbb{C}\mathbf{H}(\mathfrak{S}_{f_2}, q^2)$, respectively. Here, $\{\lambda; \lambda_1, \lambda_2 \}$ gives the number of times (or multiplicity) that the irreducible representation $[\lambda_1] \otimes [\lambda_2]$ of $\mathbb{C}\mathbf{H}(\mathfrak{S}_{f_1}, q^2) \times \mathbb{C}\mathbf{H}(\mathfrak{S}_{f_2}, q^2)$ appears in the decomposition of $[\lambda]$. The entries of the matrix transforming between split and standard basis are called \emph{subduction coefficients} (SDCs). 
\par 
Let $[\lambda_1] \otimes [\lambda_2]$ be a \emph{fixed} irreducible representation of $\mathbb{C}\mathbf{H}(\mathfrak{S}_{f_1}, q^2) \times \mathbb{C}\mathbf{H}(\mathfrak{S}_{f_2}, q^2)$ in $[\lambda] \downarrow \mathbb{C}\mathbf{H}(\mathfrak{S}_{f_1}, q^2) \times \mathbb{C}\mathbf{H}(\mathfrak{S}_{f_2}, q^2)$ and $| \lambda_1, \lambda_2 ; m_1, m_2 \rangle_{q,\eta}$  a generic vector of the split basis (where $m_1$ and $m_2$ are standard Young tableaux with Young diagram $\lambda_1$ and $\lambda_2$ respectively,  and $\eta$ is the multiplicity label). We may expand such vectors in terms of the standard basis vectors $|\lambda ; m \rangle_q $ of $[\lambda]$:
\begin{equation}
| \lambda_1, \lambda_2 ; m_1, m_2 \rangle_{q,\eta} = \sum_m \ |\lambda ; m \rangle_q \langle \lambda ; m | \lambda_1, \lambda_2 ; m_1, m_2 \rangle_{q,\eta}.
\nonumber 
\end{equation}
Thus $\langle \lambda ; m | \lambda_1, \lambda_2 ; m_1, m_2 \rangle_{q,\eta}$ represent the SDCs of $[\lambda] \downarrow [\lambda_1] \otimes [\lambda_2]$ with given multiplicity label $\eta$ and satisfy the following unitary conditions:
\begin{subequations} \label{orton}
\begin{align}
\sum_m \ \langle \lambda ; m | \lambda_1, \lambda_2 ; m_1, m_2 \rangle_{q,\eta} \ \langle \lambda ; m | \lambda_1, \lambda'_2 ; m_1, m'_2 \rangle_{q,\eta'} = \delta_{\lambda_2 \lambda'_2} \delta_{m_2 m'_2} \delta_{\eta \eta'}
\label{orton1} \\
\sum_{\lambda_2 m_2 \eta} \ \langle \lambda ; m | \lambda_1, \lambda_2 ; m_1, m_2 \rangle_{q,\eta} \ \langle \lambda ; m' | \lambda_1, \lambda_2 ; m_1, m_2 \rangle_{q,\eta} = \delta_{m m'}.
\label{orton2} 
\end{align}
\end{subequations}
\par
The explicit action of the generators $g_i$ ($i \neq f_1$ because $g_{f_1}$ is not a generator of $\mathbb{C}\mathbf{H}(\mathfrak{S}_{f_1}, q^2) \times \mathbb{C}\mathbf{H}(\mathfrak{S}_{f_2}, q^2)$) on the elements of the split basis directly follows from (\ref{actstd}). In fact we have
\begin{equation}
g_i | \lambda_1, \lambda_2 ; m_1, m_2 \rangle_{q,\eta}=
\left\{
\begin{array}{cc}
(g_i |\lambda_1; m_1 \rangle_q) \otimes |\lambda_2; m_2 \rangle_q & \text{if $1 \le i < f_{1}$ } \\
|\lambda_1; m_1 \rangle_q \otimes (g_{i-f_1}|\lambda_2; m_2 \rangle_q) & \text{if $f_1 < i < f$}
\end{array}
\right.
\label{actsplit}
\end{equation} 
Then, from (\ref{actstd}) applied to the standard basis vectors of $[\lambda_1]$ and  $[\lambda_2]$ respectively, we have the action of the generators of $\mathbb{C}\mathbf{H}(\mathfrak{S}_{f_1}, q^2) \times \mathbb{C}\mathbf{H}(\mathfrak{S}_{f_2}, q^2)$ on the basis vectors $|\lambda_1; m_1 \rangle_q \otimes |\lambda_2; m_2 \rangle_q$.
\par 
Following Pan and Chen~\cite{pan3}, we can now construct a matrix in such a way that the SDCs are the components of its kernel basis vectors. From (\ref{actsplit}), for $1\leq i < f_1$, we get (for semplicity, we omit the multiplicity label $\eta$)
\begin{equation}
\langle \lambda ; m | g_i | \lambda_1, \lambda_2; m_1, m_2 \rangle_q = \langle \lambda ; m | (g_i | \lambda_1; m_1 \rangle_q) \otimes | \lambda_2; m_2 \rangle_q
\nonumber
\end{equation}
and, writing $| \lambda_1, \lambda_2; m_1, m_2 \rangle_q$ and $g_i | \lambda_1; m_1 \rangle_q$ in the bases of $\mathbb{C}\mathbf{H}(\mathfrak{S}_{f_1}, q^2)$  and $\mathbb{C}\mathbf{H}(\mathfrak{S}_{f_2}, q^2)$ respectively, the previous relation becomes
\begin{multline}
\sum_s \ \langle \lambda; m | g_i | \lambda; s \rangle_q \langle \lambda ; s | \lambda_1, \lambda_2 ; m_1, m_2 \rangle_q  = \\ \sum_t \ \langle \lambda_1; t | g_i | \lambda_1; m_1  \rangle_q \langle \lambda ; m | \lambda_1, \lambda_2 ; t , m_2 \rangle_q  
\label{lem1}
\end{multline}
In an analogous way, for $f_1<i<f$, we get
\begin{multline}
\sum_s \ \langle \lambda; m | g_i | \lambda; s \rangle_q \langle \lambda ; s | \lambda_1, \lambda_2 ; m_1, m_2 \rangle_q  =  \\ \sum_t \ \langle \lambda_2; t | g_{i-f_1} | \lambda_2; m_2  \rangle_q \langle \lambda ; m | \lambda_1, \lambda_2 ; m_1 , t \rangle_q.
\label{lem2}
\end{multline}
\par 
Then, once we know the explicit action of the generators of $\mathbb{C}\mathbf{H}(\mathfrak{S}_{f_1}, q^2) \times \mathbb{C}\mathbf{H}(\mathfrak{S}_{f_2}, q^2)$ on the standard basis, (\ref{lem1}) and (\ref{lem2}) (written for $l \in \{1, \ldots, n_1-1, n_1+1, \ldots, n-1 \}$ and all standard Young tableaux $m$, $m_1$, $m_2$ with Young diagrams $\lambda$, $\lambda_1$ and $\lambda_2$ respectively) define a linear equation system of the form:
\begin{equation}
\Omega (\lambda; \lambda_1, \lambda_2) \ \chi = 0  
\label{subdeq}
\end{equation}
where $\Omega (\lambda; \lambda_1, \lambda_2)$ is the \emph{subduction matrix} and $\chi$ is a vector with components given by the SDCs of $[\lambda] \downarrow [\lambda_1] \otimes [\lambda_2]$. 
\par 
Denoting by $f^{\lambda}$, $f^{\lambda_1}$ and $f^{\lambda_2}$ the dimensions of the irreps $[\lambda]$, $[\lambda_1]$ and $[\lambda_2]$ respectively (which are given by the known \emph{hook formula}), (\ref{subdeq}) is a linear equation system with $ f^{\lambda} f^{\lambda_1} f^{\lambda_2}$ unknowns (the SDCs) and $(f-2) f^{\lambda} f^{\lambda_1} f^{\lambda_2}$ equations. Thus $\Omega (\lambda; \lambda_1, \lambda_2)$ is a rectangular $(f-2) f^{\lambda} f^{\lambda_1} f^{\lambda_2} \times f^{\lambda} f^{\lambda_1} f^{\lambda_2}$ matrix. 
Using the explicit action of $g_i$ given by (\ref{actstd}), we see that all equations of (\ref{subdeq}) have the form
\begin{subequations}\label{sueq}
\begin{multline}
\alpha_{i,q}(m;m_1) \langle \lambda;  m | \lambda_1, \lambda_2 
; m_1, m_2 \rangle_q - \beta_{i,q}(m)  \langle \lambda ;  
g_i(m) | \lambda_1, \lambda_2 ; m_1, m_2 \rangle_q + \\ 
\beta_{i,q}(m_1) \langle \lambda ;  m| \lambda_1 , \lambda_2 ; g_i(m_1),m_2 \rangle_q = 0 \ \ \ \ \ \ \  \text{if $1 \leq i < f_1 $ }
\label{sueq1}
\end{multline}
\begin{multline}
\alpha_{i,q}(m,m_2) \langle \lambda;  m | \lambda_1, \lambda_2 ; m_1, m_2 \rangle_q - \beta_{i,q}(m)  \langle \lambda ;  g_i(m) | \lambda_1, \lambda_2 ; m_1, m_2 \rangle_q + \\
\beta_{i-f_1,q}(m_2) \langle \lambda ;  m| \lambda_1 , \lambda_2 ; m_1,g_{i-f_1}(m_2) \rangle_q = 0 \ \ \ \ \ \ \  \text{if $f_1<i<f$ }
\label{sueq2}
\end{multline}
\end{subequations}
where 
\begin{equation}
\alpha_{i,q}(s,t) =
\left\{ 
\begin{array}{cc}
\frac{q^{d_i(t)}}{[d_i(t)]_q} - \frac{q^{d_i(s)}}{[d_i(s)]_q} & \text{if $1 \leq i < f_1 $ }  \\   
\frac{q^{d_{i-f_1}(t)}}{[d_{i-f_1}(t)]_q} - \frac{q^{d_i(s)}}{[d_i(s)]_q} &
\text{if $f_1 < i < f $ }
\end{array}
\right. .
\nonumber
\end{equation}
In the next section, it will be useful the following 
\begin{remark}
Let us suppose that $1 \leq i < f_1$ and let $\bar{m}$ be the subtableau obtained from $m$ by considering the first $f_1$ boxes, and let $\bar{m}_1$ be one standard tableau, with Young diagram $\lambda$, obtained from $m_1$ by adding $f-f_1$ boxes. By observing that
\begin{equation}
\alpha_{i,q}(\bar{m}_1,\bar{m})=- \alpha_{i,q}(m,m_1), \ \  \beta_{i,q}(\bar{m})=\beta_{i,q}(m) , \ \  \beta_{i,q}(\bar{m}_1)=\beta_{i,q}(m_1),
\nonumber
\end{equation}
equations (\ref{sueq1}), written for such tableaux, become
\begin{multline}
\alpha_{i,q}(m,m_1) \langle \lambda;  \bar{m}_1 | \lambda_1, \lambda_2 ; \bar{m}, m_2 \rangle_q - \beta_{i,q}(m)  \langle \lambda ;  g_i(\bar{m}_1) | \lambda_1, \lambda_2 ; \bar{m}, m_2 \rangle_q + \\
\beta_{i,q}(m_1) \langle \lambda ;  \bar{m}_1| \lambda_1 , \lambda_2 ; g_i(\bar{m}),m_2 \rangle_q = 0  \ \ \ \ \ \ \ \ \ \ \ \ \ \ \ \ \ \ \ \ \ \ 
\label{sueq2a}
\nonumber
\end{multline}
which imply that \emph{$\langle \lambda;  m | \lambda_1, \lambda_2 ; m_1, m_2 \rangle_q$ is a symmetric function of $d_i(m)$ and $d_i(m_1)$} (note that a standard Young tableau is biunivocally determined by its axial distances and that one has $d_i(\bar{m})=d_i(m)$ and $d_i(\bar{m}_1)=d_i(m_1)$). In general, it is not true for $f_1<i<f$.
\end{remark}
\section{Selection and identity rules}
\subsection{Crossing and bridge pairs of standard tableaux}
Let $\lambda$ be a Young diagram relative to a partition of $f$ and $(m,m')$ a pair of standard Young tableaux with the same diagram $\lambda$. Furthermore, we denote by $d_k(m)$ the usual \emph{axial distance} between the numbers $k$ and $k+1$ in the tableau $m$. \\
If $m \neq m'$, we name \emph{cut} the minimum $i \in \{1, \ldots, f-1\}$ such that $d_i(m) \neq d_i(m')$. It is useful to give the following definitions:
\begin{definition}
We say that $(m,m')$ is a \emph{crossing pair} of standard Young tableaux if there exists $i \in \{1, \ldots, f-1 \}$ such that one of the following cases is verified: 
\begin{enumerate}
\item $d_i(m) \neq d_i(m')$, $g_i(m)\neq m$ and $ g_i(m') \neq m'$;
\item $d_i(m) \neq d_i(m')$, $g_i(m) = m$ and $g_i(m') = m'$.
\end{enumerate}
We call \emph{separation} for $(m,m')$ the minimum $i$ where one of the previous cases occurs.  
\label{cros}
\end{definition}
\begin{definition}
We say that $(m,m')$ is a \emph{bridge pair} of standard Young tableaux if it is \emph{not} a crossing pair, i.e.  
for all $i \in \{1, \ldots, f-1 \}$ one of the following cases is verified: 
\begin{enumerate}
\item $d_i(m)=d_i(m')$;
\item $g_i(m)=m$ and $ g_i(m') \neq m'$;
\item $ g_i(m) \neq m$ and $g_i(m')=m'$.
\end{enumerate}
\label{brid} 
\end{definition}
Then, the following lemma holds
\begin{lemma}
Let $(m, m')$ be a bridge pair with $m \neq m'$ and let $\bar{i}$ be the relative cut. Let us consider the application defined by
\begin{equation}
g_{\bar{i}}(m,m')= (g_{\bar{i}}(m),g_{\bar{i}}(m')).
\label{bridge}
\end{equation}
Then, by iteratively applying (\ref{bridge}), we always obtain a crossing pair.
\label{ponti}
\end{lemma}
\begin{proof}
We first observe that, after one application of $g_{\bar{i}}$ on $(m,m')$, we have the following situation
\begin{equation}
\left\{ 
\begin{array}{cc}
d_j(g_{\bar{i}}(m)) = d_j(g_{\bar{i}}(m')) & \text{if $j \notin \{\bar{i}-1,\bar{i},\bar{i}+1\}$}  \\
d_j(g_{\bar{i}}(m)) = d_{j}(g_{\bar{i}}(m')) + d_{j+1}(g_{\bar{i}}(m')) & \text{if $j=\bar{i}-1$ }  \\
d_j(g_{\bar{i}}(m)) = - d_j(g_{\bar{i}}(m')) & \text{if $j=\bar{i}$ } \\
d_j(g_{\bar{i}}(m)) = d_{j-1}(g_{\bar{i}}(m')) + d_{j}(g_{\bar{i}}(m')) & \text{if $j=\bar{i}+1$ }   
\end{array}
\right.
\nonumber 
\end{equation}
thus $g_{\bar{i}}(m,m')$ has cut in $\bar{i}-1$ because obviously $d_{\bar{i}}(g_{\bar{i}}(m')) \neq 0$. \\
Then, at each step of the iteration of (\ref{bridge}), two cases may occur:
\begin{enumerate}
\item $g_{\bar{i}}(m,m')$ is a crossing pair and we have the assertion. 
\item $g_{\bar{i}}(m,m')$ is a bridge pair with cut in $\bar{i}-1$. 
\end{enumerate}
If case $(i)$ never occurs, after $\bar{i}-1$ iterations we should reach a bridge pair $(\tilde{m},\tilde{m}')$ with cut $i=1$. But $(\tilde{m},\tilde{m}')$ always is a crossing pair because $g_1(\tilde{m})=\tilde{m}$ and $g_1(\tilde{m}')=\tilde{m}'$ for each standard Young tableaux $\tilde{m}$ and $\tilde{m}'$.  
\end{proof}
\subsection{A selection rule}
Let $m$, $m_1$ and $m_2$ be three standard Young tableaux with $f$, $f_1$ and $f_2$ boxes such that $f_1 + f_2 = f$ and with Young diagrams $\lambda$, $\lambda_1$ and $\lambda_2$, respectively. 
We denote by $m^{(f_1)}$ the standard Young tableau obtained from $m$ by removing the boxes with numbers $f_1+1, \ldots, f$.
\begin{lemma}
If $(m^{(f_1)}, m_1)$ is a crossing pair of standard Young tableaux, then the subduction coefficient $\langle \lambda; m | \lambda_1, \lambda_2; m_1, m_2 \rangle_{q,\eta}$ for the reduction $[\lambda] \downarrow [\lambda_1]\otimes [\lambda_2] $ of $\mathbb{C} \mathbf{H}(\mathfrak{S}_f, q^2) \downarrow \mathbb{C}\mathbf{H}(\mathfrak{S}_{f_1}, q^2) \times \mathbb{C}\mathbf{H}(\mathfrak{S}_{f_2}, q^2) $ vanishes for all multiplicity labels $\eta$.
\label{incroci}
\end{lemma}
\begin{proof}
Let $i \in \{1, \ldots, f_1 -1\}$ be the separation of $(m^{(f_1)},m_1)$. From definition~\ref{cros}, we need to destinguish the following situations:
\begin{itemize}
\item $g_i(m) \neq m $ and $g_i(m_1) \neq m_1 $. \\
The action of the generator $g_i$ on the standard base vector $|\lambda; m \rangle_{q,\eta}$ is given by equation~(\ref{actstd}), for a fixed multiplicity label $\eta$. 
The action on the split base vector $ |\lambda_1, \lambda_2;m_1,m_2 \rangle_{q,\eta}$ is defined in an analogous way. Therefore, by using the equation~(\ref{gen3}) and the fact that $g_i = {g_i}^{\dagger}$, we get the following relation for the SDCs: 
\begin{multline}
\left(1-\frac{q^{d_i(m_1)-d_i(m)}}{[d_i(m)]_q [d_i(m_1)]_q}\right) \ \langle \lambda;m|\lambda_1, \lambda_2 ;m_1, m_2 \rangle_{q,\eta} \ + \\  - \ \frac{\beta_{i,q}(m_1) q^{-d_i(m)}}{[d_i(m)]_q} \  \langle \lambda; m |\lambda_1, \lambda_2 ; g_i(m_1), m_2 \rangle_{q,\eta} \ + \\
- \ \frac{\beta_{i,q}(m)q^{d_i(m_1)}}{[d_i(m_1)]_q} \  \langle \lambda; g_i(m)|\lambda_1,\lambda_2 ; m_1, m_2 \rangle_{q,\eta} \ + \\ - \ \beta_{i,q}(m) \beta_{i,q}(m_1) \ \langle\lambda; g_i(m)|\lambda_1, \lambda_2; g_i(m_1), m_2 \rangle_{q,\eta}  = 0.
\nonumber
\end{multline}
The previous equation, also written for the other subduction coefficients, i.e. $\langle\lambda; m |\lambda_1, \lambda_2;g_i(m_1), m_2  \rangle_{q,\eta}$, $\langle \lambda;g_i(m)|\lambda_1, \lambda_2; m_1, m_2  \rangle_{q,\eta}$ and $\langle\lambda; g_i(m)|\lambda_1, \lambda_2 ; g_i(m_1), m_2  \rangle_{q,\eta}$, by noting that $d_i(g_i(m))=-d_i(m)$ and $d_i(g_i(m_{12}))=-d_i(m_{12})$, provides the homogeneous linear system described by the matrix $\mathcal{S} = \mathcal{I} - \mathcal{M}$, where $\mathcal{I}$ is the $4 \times 4$ identity matrix and 
\begin{equation}
\mathcal{M}=
\left(
\begin{array}{cc}
\frac{q^{-d_i(m)}}{[d_i(m)]_q} & \beta_{i,q}(m)  \\
\beta_{i,q}(m) & -\frac{q^{d_i(m)}}{[d_i(m)]_q}
\end{array}
\right)
\otimes
\left(
\begin{array}{cc}
\frac{q^{d_i(m_1)}}{[d_i(m)]_q} & \beta_{i,q}(m_1)  \\
\beta_{i,q}(m_1) & -\frac{q^{-d_i(m_1)}}{[d_i(m_1)]_q}
\end{array}
\right).
\nonumber
\end{equation}
So, the SDCs under consideration belong to the kernel space of the matrix $\mathcal{S}$ (i.e. the eigenspace relative to the eigenvalue $0$) which has dimension $2$ (note that the eigenvalues of $\mathcal{S}$ are $0$, with multiplicity $2$, $1+q^2$ and $1+q^{-2}$). By remembering the remark in the previous section, we need to choose the vectors in the kernel of $S$ which are formed by symmetric functions of $d_i(m)$ and $d_i(m_1)$.
\begin{enumerate}
\item If $d_i(m) \neq - d_i(m_{1})$; \\
 it is easy to see that they are of the form
\begin{equation}
s(d_i(m), d_i(m_1)) 
\left(
\begin{array}{c}
1 \\
\frac{[d_i(m)]_q[d_i(m_1)]_q(\beta_{i,q}(m) +\beta_{i,q}(m_1))}{[d_i(m)+d_i(m_1)]_q}    \\
\frac{[d_i(m)]_q[d_i(m_1)]_q(\beta_{i,q}(m) +\beta_{i,q}(m_1))}{[d_i(m)+d_i(m_1)]_q}    \\
-1
\end{array}
\right),
\nonumber
\end{equation}
where $s(d_i(m), d_i(m_1))$ is a symmetric function. Thus
\begin{equation}
\langle \lambda;m |\lambda_1,\lambda_2 ; m_1, m_2 \rangle_{q,\eta} = - \langle \lambda; g_i(m) |\lambda_1, \lambda_2; g_i(m_1), m_2 \rangle_{q,\eta}
\label{cond1}
\end{equation} 
 and 
\begin{equation}
\langle \lambda; g_i(m)|\lambda_1, \lambda_2; m_1, m_2 \rangle_{q,\eta} = \langle \lambda;m |\lambda_1,\lambda_2; g_i(m_1), m_2 \rangle_{q,\eta}.
\label{cond2}
\end{equation}
Because $d_i(g_i(m)) = -d_i(m)$, $d_i(m) \neq - d_i(m_1) \Rightarrow d_i(g_i(m)) \neq d_i(m_1)$ and $d_i(m) \neq d_i(m_1) \Rightarrow d_i(g_i(m)) \neq - d_i(m_1)$. Therefore relation (\ref{cond2}), written for $\langle \lambda; g_i(m) | \lambda_1 , \lambda_2 ;m_1, m_2 \rangle_{q,\eta}$, yelds (remember that $g_i^2 (m)= m$)
 \begin{equation}
\langle \lambda ; m |\lambda_1, \lambda_2 ; m_1, m_2 \rangle_{q,\eta} = \langle\lambda; g_i(m)| \lambda_1, \lambda_2 ; g_i(m_1), m_2 \rangle_{q,\eta}.
 \label{cond5} 
\end{equation}
From (\ref{cond1}) and (\ref{cond5}), we get 
\begin{equation}
\langle\lambda; m|\lambda_1, \lambda_2;m_1, m_2 \rangle_{q,\eta} = 0.
\nonumber
\end{equation}
\item 
If $d_i(m) = -d_i(m_1)$; \\
the set of our symmetric SDCs is given by the vectors of the form (again, $s(d_i(m),d_i(m_1))$ represents a symmetric function)
\begin{equation}
s(d_i(m),d_i(m_1))
\left(
\begin{array}{c}
0 \\
1 \\
1 \\
0
\end{array}
\right),
\nonumber
\end{equation}
which directly implies
\begin{equation}
\langle \lambda;g_i(m)| \lambda_1, \lambda_2 ; m_1, m_2 \rangle_{q,\eta} = \langle\lambda; m|\lambda_1, \lambda_2; g_i(m_1), m_2 \rangle_{q,\eta}
\label{cond3}
\end{equation}
and
\begin{equation}
\langle \lambda; m|\lambda_1, \lambda_2 ;m_1, m_2 \rangle_{q,\eta} = \langle \lambda; g_i(m)| \lambda_1, \lambda_2; g_i(m_1), m_2 \rangle_{q,\eta} = 0.
\label{cond2bis}
\end{equation} 
\end{enumerate}
\item $g_i(m) = m$ and $g_i(m_1) = m_1$. \\
The action of the generator $g_i$ on the standard base vector $|\lambda; m \rangle_{q,\eta}$ is given by
\begin{equation}
g_i |\lambda; m \rangle_{q,\eta} = \pm q^{\pm 1} |\lambda; m \rangle_{q,\eta}
\nonumber
\end{equation}
and the action on the split base vector $ |\lambda_1, \lambda_2;m_1,m_2 \rangle_{q,\eta}$ is
\begin{equation}
g_i |\lambda_1, \lambda_2 ; m_1, m_2 \rangle_{q,\eta} = \mp q^{\mp 1}|\lambda_1,\lambda_2; m_1, m_2 \rangle_{q,\eta}.
\nonumber 
\end{equation}
Thus, from relation~(\ref{gen3}) and by $g_i={g_i}^{\dagger}$,
\begin{equation}
(q^{\pm 2} + 1) \langle \lambda;m|\lambda_1, \lambda_2;m_1,m_2 \rangle_{q,\eta} = 0
\nonumber
\end{equation}
and then
\begin{equation}
\langle \lambda;m|\lambda_1, \lambda_2;m_1,m_2 \rangle_{q,\eta} = 0.
\nonumber
\end{equation} 
\end{itemize}
\end{proof}
\begin{remark}
It is important to point out that, in the previous proof, the fact that $q \in \mathbb{R}$ is crucial because we need the Hermitianity of the $g_i$ action: $g_i={g_i}^{\dagger}$. On the other hand, such a condition guarantees the inequality $q^2+1 \neq 0$.
\end{remark}
\begin{lemma}
If $(m^{(f_1)}, m_1)$ is a bridge pair of standard Young tableaux, then the subduction coefficient $\langle \lambda; m | \lambda_1, \lambda_2; m_1, m_2 \rangle_{q,\eta}$ for the reduction $[\lambda] \downarrow [\lambda_1]\otimes [\lambda_2] $ of $\mathbb{C} \mathbf{H}(\mathfrak{S}_f, q^2) \downarrow \mathbb{C}\mathbf{H}(\mathfrak{S}_{f_1}, q^2) \times \mathbb{C}\mathbf{H}(\mathfrak{S}_{f_2}, q^2) $ vanishes for all multiplicity labels $\eta$.
\label{ponti1}
\end{lemma}
\begin{proof}
Let $i \in \{1, \ldots, f_1 -1\}$ be the cut of $(m^{(f_1)},m_1)$. For semplicity, let us suppose $g_i(m) = m $ and $g_i(m_1)\neq m_1 $. \\
The action of the generator $g_i$ on the standard base vector $|\lambda; m \rangle$ is given by
\begin{equation}
g_i |\lambda; m \rangle_{q,\eta} = \pm q^{\pm 1} |\lambda; m \rangle 
\nonumber
\end{equation}
and the action on the split base vector $ |\lambda_1, \lambda_2;m_1,m_2 \rangle_{q,\eta}$ is
\begin{multline}
g_i |\lambda_1, \lambda_2 ; m_1, m_2 \rangle_{q,\eta} = \\ \frac{q^{d_i(m_1)}}{[d_i(m_1)]_q} \ |\lambda_1,\lambda_2; m_1, m_2 \rangle_{q,\eta} \ + \ \beta_{i,q}(m_1) \ |\lambda_1,\lambda_2; g_i(m_1), m_2 \rangle_{q,\eta}
\nonumber
\end{multline} 
Because $|d_i(m_1)|\neq  1$, the previous actions imply 
\begin{equation}
\langle \lambda;m|\lambda_1, \lambda_2;m_1,m_2 \rangle_{q,\eta} = b'_{i,q}(m_1) \langle \lambda;m|\lambda_1, \lambda_2;g_i(m_1),m_2 \rangle_{q,\eta}
\nonumber
\end{equation}
with $b'_{i,q}(m_1)$ a suitable factor.\\ 
In an analogous way, the case $g_i(m) \neq m$ and $g_i(m_1)=m_1$ provides
\begin{equation}
\langle \lambda;m|\lambda_1, \lambda_2;m_1,m_2 \rangle_{q,\eta} = b''_{i,q}(m) \langle \lambda;g_i(m)|\lambda_1, \lambda_2;m_1,m_2 \rangle_{q,\eta}
\nonumber
\end{equation}
with $b''_{i,q}(m)$ another suitable factor. \\
From lemma~\ref{ponti}, by iterating the previous derivation, we may write
\begin{equation}
\langle \lambda;m|\lambda_1, \lambda_2;m_1,m_2 \rangle_{q,\eta} = b_{q} \ \langle \lambda; \bar{m} |\lambda_1, \lambda_2;\bar{m}_1,m_2 \rangle_{q,\eta} 
\nonumber
\end{equation}
with $(\bar{m}, \bar{m}_1)$ a crossing pair and $b_{q}$ a total numerical factor. But, from lemma \ref{incroci},
\begin{equation}
\langle \lambda; \bar{m} |\lambda_1, \lambda_2;\bar{m}_1,m_2 \rangle_{q,\eta} = 0,
\nonumber
\end{equation}
thus 
\begin{equation}
\langle \lambda;m|\lambda_1, \lambda_2;m_1,m_2 \rangle_{q,\eta} = 0.
\nonumber
\end{equation}
\end{proof}
It is now possible to give the following 
\begin{theorem}[Selection Rule]
If $m^{(f_1)} \neq m_1$, then the subduction coefficient $\langle \lambda; m | \lambda_1, \lambda_2; m_1, m_2 \rangle_{q,\eta}$ for the reduction $[\lambda] \downarrow [\lambda_1]\otimes [\lambda_2] $ of $\mathbb{C} \mathbf{H}(\mathfrak{S}_f, q^2) \downarrow \mathbb{C}\mathbf{H}(\mathfrak{S}_{f_1}, q^2) \times \mathbb{C}\mathbf{H}(\mathfrak{S}_{f_2}, q^2) $ vanishes for all multiplicity labels $\eta$.
\label{selec}
\end{theorem}
\begin{proof}
If $m^{(f_1)} \neq m_1$, then $(m^{(f_1)},m_1)$ is a crossing or bridge pair of standard Young tableaux. Thus, the proof follows from lemma~\ref{incroci} and lemma~\ref{ponti1}. 
\end{proof}
\subsection{An equivariance relation}
We now show another simple proposition that goes further on reducing the number of unknown SDCs in the subduction problem of type A Iwahori-Hecke algebras. With the same notation of the previous subsection, it is possible to give the following 
\begin{theorem}[Identity Rule]
If $m^{(f_1)}=m$, then 
\begin{equation}
\langle\lambda; m|\lambda_1, \lambda_2; m_1, m_2 \rangle_{q,\eta} = \langle \lambda; g_i(m)| \lambda_1,\lambda_2; g_i(m_1), m_2 \rangle_{q,\eta}
\nonumber
\end{equation}
for all $\leq i\leq f_1 -1$.
\label{equiv}
\end{theorem}
\begin{proof}
Because $m^{(f_1)}=m$, we have $d_i(m)=d_i(m_1)=d_i$ and $\beta_{i,q}(m)=\beta_{i,q}(m_1)=\beta_{i,q}$ for all $1\leq i \leq f_1-1$. Again, we distinguish two cases: 
\begin{itemize}
\item $|d_i|\neq 1$.\\  It is straightforward that the matrix $\mathcal{S}$, written in the previous subsection, has the kernel space which is generated by the vectors
\begin{equation}
\left(
\begin{array}{c}
1 \\
0 \\
0 \\
1
\end{array}
\right), \ \ \ 
\left(
\begin{array}{c}
\frac{[2 d_i]_q}{{[d_i]_q}^2\beta_{i,q}} \\
1 \\
1 \\
0
\end{array}
\right)
\nonumber
\end{equation}
therefore we have
\begin{equation}
\langle\lambda;  g_i(m)|\lambda_1,\lambda_2;m_1,m_2 \rangle_{q,\eta} = \langle\lambda; m| \lambda_1,\lambda_2;g_i(m_1), m_2 \rangle_{q,\eta}
\label{cond2tris}
\end{equation}
and
\begin{multline}
\langle \lambda;m|\lambda_1,\lambda_2; m_1,m_2 \rangle_{q,\eta} = \langle \lambda;g_i(m)|\lambda_1,\lambda_2; g_i(m_1), m_2 \rangle_{q,\eta} + \\
+ \frac{[2 d_i]_q}{{[d_i]_q}^2 \beta_{i,q}} \langle\lambda;  g_i(m)| \lambda_1, \lambda_2; m_1, m_2\rangle_{q,\eta}.
\label{cond4}
\end{multline}
Because $d_i(g_i(m))= - d_i(m)=-d_i(m_1)$, (\ref{cond2bis}) becomes  
\begin{equation}
\langle \lambda; g_i(m)| \lambda_1, \lambda_2; m_1, m_2 \rangle_{q,\eta} = 0
\nonumber
\end{equation}
and, from (\ref{cond4}),
\begin{equation}
\langle\lambda; m|\lambda_1, \lambda_2; m_1, m_2 \rangle_{q,\eta} = \langle \lambda; g_i(m)| \lambda_1,\lambda_2; g_i(m_1), m_2 \rangle_{q,\eta}.
\nonumber
\end{equation}
\item $d_i = \pm 1$.\\ This is a trivial case because $g_i(m)=m$ and $g_i(m_{1})=m_{1}$. Thus, of course, 
\begin{equation}
\langle \lambda;m| \lambda_1, \lambda_2; m_1,m_2 \rangle_{q,\eta} = \langle \lambda; g_i(m)| \lambda_1,\lambda_2; g_i(m_1), m_2 \rangle_{q,\eta}.
\nonumber
\end{equation}
\end{itemize}
\end{proof}
\begin{remark}
From theorems~\ref{selec} and~\ref{equiv}, the general structure of the SDCs follows: 
$$\langle \lambda;m| \lambda_1, \lambda_2; m_1,m_2 \rangle_{q,\eta} = c_{q,\eta}(\lambda, \lambda_1, \lambda_2; m, m_2) \ \delta_{m_1,m^{(f_1)}}.$$
In other words, each SDC is not dependent on the standard Young tableau $m_1$ and it vanishes if $m^{(f_1)}\neq m_1$.
\end{remark}
\section{Discussion}
From the previous theorems, we get that only a smaller number of the $f^{\lambda} f^{\lambda_1} f^{\lambda_2}$ unknowns for the primal subduction problem really needs to be evaluated. In fact, if $f^{\lambda \backslash \lambda_1}$ represents the number of \emph{skew-tableaux} with diagram $\lambda \backslash \lambda_1$ (which is equal to the number of tableaux $m$ with Young diagram $\lambda$ such that $m^{(f_1)}=m_1$, where $m_1$ is a standard Young tableau with diagram $\lambda_1$), the number of not vanishing SDCs is given by  $f^{\lambda \backslash \lambda_1} f^{\lambda_1} f^{\lambda_2}$. 
\par
In table~\ref{tabmult} we deal with some subduction cases, the relative multiplicity, the number of unknowns involved in the primal linear equation system and the effective number of needed SDCs, after the application of the selection and identity rules. It is evident the drastic reduction of the number of unknowns for the subduction problem.   
\begin{table}
\begin{center}
\begin{tabular}{cccc}
\toprule
$[\lambda]\downarrow [\lambda_1]\otimes [\lambda_2]$ & $\{\lambda; \lambda_1, \lambda_2 \}$ & $f^{\lambda} f^{\lambda_1} f^{\lambda_2}$& $f^{\lambda/\lambda_1} f^{\lambda_2}$ \\
\midrule
$[4,2]\downarrow [2,1]\otimes [2,1]$ & $1$ & $36$ & $6$ \\
$[3,2,1]\downarrow [2,1]\otimes [2,1]$ & $2$ & $64$  & $12$  \\
$[4,2,1]\downarrow [3,1]\otimes [2,1] $ & $2$ & $210$ & $12$ \\
$[4,3,2]\downarrow [3,2]\otimes [3,1]$ & $2$ & $2520$ & $36$ \\
$[4,3,2,1]\downarrow [3,2,1]\otimes [3,1]$ & $3$ & $36864$ & $72$ \\
$[5,4,3,2]\downarrow [4,3,2]\otimes [3,2]$ & $3$ & $40360320$ & $300$ \\
$[5,4,3,2,1]\downarrow [4,3,2,1]\otimes [4,1]$ & $4$ & $899678208$ & $480$ \\
\bottomrule
\end{tabular}
\end{center}
\caption{Some examples of subduction with the relative multiplicity, the primal number of involved SDCs and the final number of not vanishing SDCs which really need to be evaluated.}
\label{tabmult}
\end{table}
\par
Following~\cite{Chilla}, we obtain that a \emph{reduced} version of the subduction graph is sufficient to provide the required transformation coefficients. Such a graph is obtained by the action of the $g_i$ transformations, with $i \in \{f_1 +1, \ldots, f-1\}$, on the tableaux corresponding to the not vanishing SDCs. On the other hand, the $g_i$ transformations with $i \in \{1, \ldots, f_1 - 1 \}$, by the identity rule, can be not considered.
\par
As an example of computation thanks to the results of this paper, we could study the first multiplicity-three case for the subduction problem in type A Iwahori-Hecke algebras which accours in $[4,3,2,1] \downarrow [3,2,1] \otimes [3,1]$ of $ \mathbb{C}\mathbf{H}(\mathfrak{S}_{10}, q^2) \downarrow  \mathbb{C}\mathbf{H}(\mathfrak{S}_{6}, q^2) \times \mathbb{C}\mathbf{H}(\mathfrak{S}_{4}, q^2) $ and which was never solved before. This higher multiplicity example of subduction is particularly significant to understand how choose the multiplicity separation in a general way.
\par 
Because the representation $[4,3,2,1]$ has dimension $f^{\lambda}=768$, $[3,2,1]$ has dimension $f^{\lambda_1}=16$ and $[3,1]$ dimension $f^{\lambda_2}=3$, we have $f^{\lambda} f^{\lambda_1} f^{\lambda_2}=36864$ SDCs to evaluate. Many of such coefficients vanish via the selection rule provided in the previous section. On the other hand, from the identity rule, we only need to determine the SDCs corresponding to the skew-tableaux of shape $[4,3,2,1]/[3,2,1]$, i.e. we have $f^{\lambda/\lambda_1} f^{\lambda_2} =72$ unknowns. We may use the Yamanouchi convention~\cite{Chenbook} to fix the phase freedom by choosing the \emph{first non-zero} SDC to be \emph{positive}. From subduction graph~\cite{Chilla} and by, for example, a suitable Mathematica program, we may then generate the homogeneous linear sistem required to obtain the SDCs and evaluate the kernel of the subduction matrix. The solution space has dimension $3$ (multiplicity) and the SDCs are orthonormalized in such a way that the conditions (\ref{orton1}) and (\ref{orton2}) hold.
\par 
We finally observe that by coniugation of the previous case, i.e. by considering $[4,3,2,1] \downarrow [3,2,1] \otimes [2,1,1]$, we have another multiplicity three case for the subduction problem. The new SDCs are related to the previous ones by the phase factors of the Yamanouchi basis~\cite{Butler} for the irrep $[\lambda]$, $[\lambda_1]$ and $[\lambda_2]$.


\end{document}